\documentclass[letterpaper, 10 pt, conference]{ieeeconf}  

\IEEEoverridecommandlockouts                              

\usepackage{amsmath}
\usepackage{amssymb}

\usepackage{amsthm}

\usepackage{graphicx}
\usepackage{booktabs}
\let\labelindent\relax
\usepackage{enumitem}
\usepackage{pifont}
\usepackage[tight,footnotesize]{subfigure}
\usepackage{cite}
\usepackage{algorithm}
\usepackage{algorithmic}
\usepackage[bookmarks=true]{hyperref}
\hypersetup{
	colorlinks=true,
	linkcolor=blue,
	citecolor=red,
}

\theoremstyle{remark}
\newtheorem{remarkx}{Remark}
\newenvironment{remark}
{\pushQED{\qed}\remarkx}
{\popQED\endremarkx}

\theoremstyle{definition}

\newtheorem{assump}{Assumption}
\newtheorem*{problem*}{Problem}
\newtheorem{problem}{Problem}

\theoremstyle{plain}
\newtheorem{theorem}{Theorem}
\newtheorem{lemma}{Lemma}

\newtheorem{prop}{Proposition}

\newcommand{\dist}{{\rm dist}}

\newcommand{\defeq}{:=} 
\newcommand{\matr}[1]{\begin{bmatrix} #1 \end{bmatrix}}

\newcommand{\transpose}[1]{#1^\top}
\newcommand{\norm}[1]{\left\lVert#1\right\rVert}

\newcommand{\set}[1]{\mathcal{#1}}
\newcommand{\setp}{\mathcal{P}}
\newcommand{\setc}{\mathcal{C}}

\newcommand{\mbr}[1][{}]{\mathbb{R}^{#1}}    

\newcommand{\dw}[1]{\frac{{\rm d} #1}{{\rm dw}}}

\newcommand{\cmmnt}[1]{}
\newcommand{\chiup}{\raisebox{2pt}{$\chi$}}
\newcommand{\vf}{\chiup}

\newcommand{\proj}[1]{{#1}^{\mathrm{prj}}}

\newcommand{\hgh}[1]{{#1}^{\mathrm{hgh}}}
\newcommand{\phy}[1]{{#1}^{\mathrm{phy}}}

\title{\LARGE \bf
	Vector Field Guided Path Following Control: \\ Singularity Elimination and Global Convergence
}

\author{Weijia Yao, H{\' e}ctor Garcia de Marina and Ming Cao 
	\thanks{Weijia Yao and Ming Cao are with ENTEG, University of Groningen, the Netherlands. Hector Garcia de Marina is with Universidad Complutense de Madrid, 28040 Madrid, Spain. {\tt\small \{w.yao,m.cao\}@rug.nl} {\tt\small hgdemarina@gmail.com}. Weijia Yao is funded by the China Scholarship Council (CSC). The work of H{\' e}ctor Garcia de Marina is supported by the grant Atraccion de Talento 2019-T2/TIC-13503 from the Government of the Autonomous Community of Madrid.  }%
}

\begin{document}

	\maketitle
	\thispagestyle{empty}
	\pagestyle{empty}

	\begin{abstract}
	Vector field guided path following (VF-PF) algorithms are fundamental in robot navigation tasks, but  may not deliver the desirable performance when robots encounter singular points where the vector field becomes zero. The existence of singular points prevents the global convergence of the vector field's integral curves to the desired path. Moreover, VF-PF algorithms, as well as most of the existing path following algorithms, fail to enable following a self-intersected desired path. In this paper, we show that such failures are fundamentally related to the mathematical topology of the path, and that by ``stretching" the desired path along a virtual dimension, one can remove the topological obstruction. Consequently, this paper proposes a new guiding vector field defined in a higher-dimensional space, in which self-intersected desired paths become free of self-intersections; more importantly, the new guiding vector field does \emph{not} have any singular points, enabling the integral curves to converge globally to the ``stretched'' path. We further introduce the \emph{extended dynamics} to retain this appealing global convergence property for the desired path in the original lower-dimensional space. Both simulations and experiments are conducted to verify the theory.
	\end{abstract}
	
	
	\section{Introduction}
	The path following control problem is to find suitable control laws for a mobile vehicle to converge to and traverse along a prescribed geometric path, which is, mathematically, a one-dimensional manifold. Path following is an essential function for many mobile robots \cite{siciliano2010robotics}, and recently related new applications have emerged, e.g., to probe atmospheric phenomena by drones \cite{lacroix2016fleets}. Among many path following algorithms \cite{Sujit2014,cai2009information,lu2014information}, vector field guided path following (VF-PF) algorithms are promising as they can achieve low path-following errors while requiring small control effort \cite{Sujit2014}. In VF-PF algorithms, a vector field is carefully designed such that its integral curves converge to the desired path \cite{nelson2007vector,lawrence2008lyapunov,Goncalves2010}. Such a vector field is also known as a \emph{guiding vector field} \cite{Y.A.2017}, since the desired velocity at each point of the field guides the robot. 
	
	A variety of VF-PF algorithms exist in the literature \cite{nelson2007vector,lawrence2008lyapunov,Goncalves2010, Y.A.2017,liang2016combined,nelson2006vector}. However, the existence of singular points where a vector field becomes zero 
	compromises the global convergence to the desired path. It also complicates the analysis of the algorithms \cite{Y.A.2017,Goncalves2010,lawrence2008lyapunov}. Sometimes singular points are unavoidable. For example, if a 2D desired path is a simple closed curve, then at least one singular point always exists within the region enclosed by the desired path, as disclosed by the Poicar\'e-Bendixson theorem \cite[Corollary 2.1]{khalil2002nonlinear}. The hairy ball theorem also guarantees the existence of singular points of continuous vector fields defined on the sphere $\mathbb{S}^2$ \cite[Theorem 13.32]{lee2010topologicalmanifolds}. Therefore, for example, if the desired path is a circle, a robot cannot start and then escape from the center, which is a singular point of the 2D guiding vector field \cite{Y.A.2017,Goncalves2010}. In addition, the normalization of the vector field at this point is not well-defined. To the best of our knowledge, few studies deal with singular points of guiding vector fields. Some studies (e.g., \cite{Goncalves2010,Rezende2018}) assume that the singular points are repulsive to simplify the analysis. This assumption is dropped in \cite{Y.A.2017} for the 2D case, but the paper concludes that the extensibility of the integral curves might be finite if there are singular points. In general, path following algorithms only guarantee local convergence to the desired path (e.g., \cite{fossen2003line,li2009design,park2007performance}).
	
	Following self-intersected desired paths is not achievable for existing VF-PF algorithms, since the crossing points of the desired path are also singular points of the vector field; hence, no guiding directions are available at the intersections. Thus, a robot gets stuck at the crossing points on the desired path (see Fig. \ref{fig:eightvf}). For this reason, the algorithms in \cite{Goncalves2010,Y.A.2017,nelson2007vector,lawrence2008lyapunov} are ineffective simply due to the violation of the (implicit) assumption: no singular points are allowed to be on the desired path. In fact, many other path following algorithms are only applicable to simple desired paths such as a straight line or a circle, or have not addressed the problem of following self-intersected paths (e.g. \cite{nelson2007vector}). For example, the classic line-of-sight (LOS) method \cite{fossen2003line,fossen2014line,rysdyk2006unmanned} fails in this case as there is no unique projection point in the vicinity of a crossing point of the desired path. 
	
	When obstacle-populated environments are considered, a closely related line of research is feedback motion planning, where a feedback plan is equivalently a vector field defined on some configuration space \cite[Chapter 8]{lavalle2006planning}. One of the most influential feedback motion planner is based on a navigation function \cite{rimon1992exact,dimarogonas2003decentralized,loizou2002closed}, but the derived vector field always has undesirable singular points, due to the topology of the configuration space that is punctured by obstacles \cite{rimon1992exact}. Therefore, eliminating these singular points in the \emph{same} configuration space is infeasible. Approximate cell decomposition and probabilistic roadmap methods are also used as general path following solutions \cite{cai2009information,lu2014information,zhang2009global,ko2014randomized}.

	\textbf{Our contributions}: As explained above, undesirable singular points of a vector field are common in many existing studies \cite{Goncalves2010,lawrence2008lyapunov,nelson2007vector,rimon1992exact}. To tackle the above issues caused by singular points, we propose a new idea to change the topology of the desired path, and transform the guiding vector field into a singularity-free one in a \emph{higher dimensional} configuration space. 
	Consequently, we rigorously guarantee the global convergence of the vector field's integral curves to the desired path, self-intersected or not. Therefore, we substantially improve the performance of the existing VF-PF algorithms. Moreover, we show that our proposed algorithm is a combined extension of both existing VF-PF and trajectory tracking algorithms.  Note that our approach can be easily generalized for desired paths in a higher-dimensional space \cite{yao2020singularity}. In addition, combining with the idea from \cite{yao2019integrated}, our approach can generate \emph{singularity-free} vector fields for robot navigation in obstacle-populated environments, removing the undesirable singularity in \cite{rimon1992exact}.
	
	This paper is organized as follows. Section \ref{sec2} introduces the 2D vector field and the motivation of our work. Then in Section \ref{sec3}, the 3D vector field and the problem definition are presented. The main results are elaborated in Section \ref{sec4}. In Section \ref{sec5}, an experiment is described, and the comparison with a trajectory tracking algorithm is presented in Section \ref{sec6}. Section \ref{sec7} concludes the paper. 
	
	\section{Motivation} \label{sec2}
	In the VF-PF problem, the desired path to follow is usually described by the zero-level set of a sufficiently smooth function. In particular, in the 2D Euclidean space $\mbr[2]$, the desired path is described below \cite{Y.A.2017}:
	\begin{equation} \label{eqpath2d}
	\overline{\setp} = \{ (x,y) \in \mbr[2] : \overline{\phi}(x,y) = 0\},
	\end{equation} 
	where $\overline{\phi}:\mbr[2] \to \mbr[]$, called the \emph{surface function}, is twice continuously differentiable to guarantee the existence and uniqueness \cite[Theorem 3.1]{khalil2002nonlinear} of the integral curves of the guiding vector field introduced in the sequel. Note that this description of the desired path without any parametrization is common in the field of VF-PF control \cite{yao2018cdc,Goncalves2010,Rezende2018,michalek2018vfo,morro2011path,do2015global,chen2011curve}. One can exploit the description (\ref{eqpath2d}) by using the absolute value of the surface function $|\overline{\phi}(p)|$ instead of the Euclidean distance $\dist(p, \overline{\setp}) \defeq \inf\{\norm{p-q} : q \in \overline{\setp}\}$ between a robot's position $p \in \mbr[2]$ and the desired path $\overline{\setp}$ to measure how far the robot is away from the desired path. For example, when $|\overline{\phi}(p)|=0$, the robot is precisely on the desired path. Thus, the 2D VF-PF problem is the design of a continuously differentiable vector field $\overline{\vf}: \mbr[2] \to \mbr[2]$ to show up on the right-hand side of the differential equation $\dot{p}(t) = \overline{\vf}(p(t))$. The design requires to satisfy two conditions:
	(a) There exists a neighborhood $\mathcal{D} \subseteq \mbr[2]$ of the desired path $\overline{\setp}$ in \eqref{eqpath2d} such that for all initial conditions $p(0) \in \mathcal{D}$, the distance $\dist(p(t), \overline{\setp})$ between the trajectory $p(t)$ and the desired path $\overline{\setp}$ approaches zero as time $t \to \infty$; that is, $\lim_{t \to \infty} \dist(p(t), \overline{\setp}) = 0$;
	(b) If a trajectory starts from the desired path, then it stays on the path for $t \ge 0$ (i.e., $p(0) \in \overline{\setp} \implies p(t) \in \overline{\setp}$ for all $t \ge 0$). 
	
	The trajectories (solutions) to the differential equation $\dot{p}(t) = \overline{\vf}(p(t))$ are called the \emph{integral curves} of the vector field $\overline{\vf}$. Therefore, we aim to find a suitable vector field $\overline{\vf}$ such that the integral curves approach the desired path eventually, and propagate (travel) along the desired path once it starts from it. We say that the integral curves of the vector field that satisfy the above two conditions assist a robot in converging to and travelling along the desired path. One example of such a vector field is given in \cite{Y.A.2017}:
	\begin{equation} \label{eqvf}
	\overline{\vf}(x,y) = E \nabla \overline{\phi}(x,y) - k \psi(\overline{\phi}(x,y)) \nabla \overline{\phi}(x,y),
	\end{equation}
	where $E=\left[\begin{smallmatrix}0 & -1 \\ 1 & 0\end{smallmatrix}\right]$ is the $90^\circ$ rotation matrix, $k$ is a positive constant, $\psi:\mbr[] \to \mbr[]$ is a strictly increasing function with $\psi(0)=0$, and $\nabla (\cdot)$ denotes the gradient of a scalar function $(\cdot)$. Note that one of the assumptions is that there are no singular points of the vector field on the desired path. However, the vector field must have at least one singular points on a self-intersected desired path as stated in the following proposition. 
	\begin{theorem}[Crossing Points are Singular Points] \label{prop1}
		Denote a crossing point of the desired path $\overline{\setp}$ in \eqref{eqpath2d} by $c \in \overline{\setp}$. Then $c$ is also a singular point of the vector field in \eqref{eqvf}; that is, $\overline{\vf}(c)=\boldsymbol{0}$.
	\end{theorem}
	\begin{proof}
		Since $c \in \overline{\setp}$, we have $\overline{\phi}(c)=0$, and thus $\chi(c) = E \nabla \overline{\phi}(c)$ in view of \eqref{eqvf}. Next we show that the gradient at the crossing point $\nabla \overline{\phi}(c)$ is zero; hence $\chi(c)=0$. Suppose, on the contrary, the gradient is nonzero. Since $\nabla \overline{\phi}$ is continuously differentiable (as $\overline{\phi} \in C^2$), the implicit function theorem \cite{giaquinta2010mathematical} concludes that there is a unique curve $\gamma: (-\delta, \delta) \to \set{U}$ in a neighborhood $\set{U}$ of $c$ satisfying $\gamma(0)=c$ and $\overline{\phi} \circ \gamma(s) =0$ for all $s \in (-\delta, \delta)$. But this contradicts the fact that $\overline{\setp}$ is self-intersected. Therefore, the gradient at the crossing point is indeed zero, resulting in $\overline{\vf}(c)=0$.
	\end{proof} 
	In Fig. \ref{fig:eightvf}, one can check that the vector field at the crossing point is zero. We can attribute this to the loss of directional information at the crossing point: there is no obvious preference regarding which direction the trajectory should go towards. Thus, it is impossible to create a 2D vector field of which the integral curves follow the self-intersected desired path. In the following sections, we propose a solution to this inherent limitation by designing a new guiding vector field defined in the higher-dimensional Euclidean space $\mbr[3]$.
	\begin{figure}[tb]
		\centering
		\includegraphics[width=0.4\linewidth]{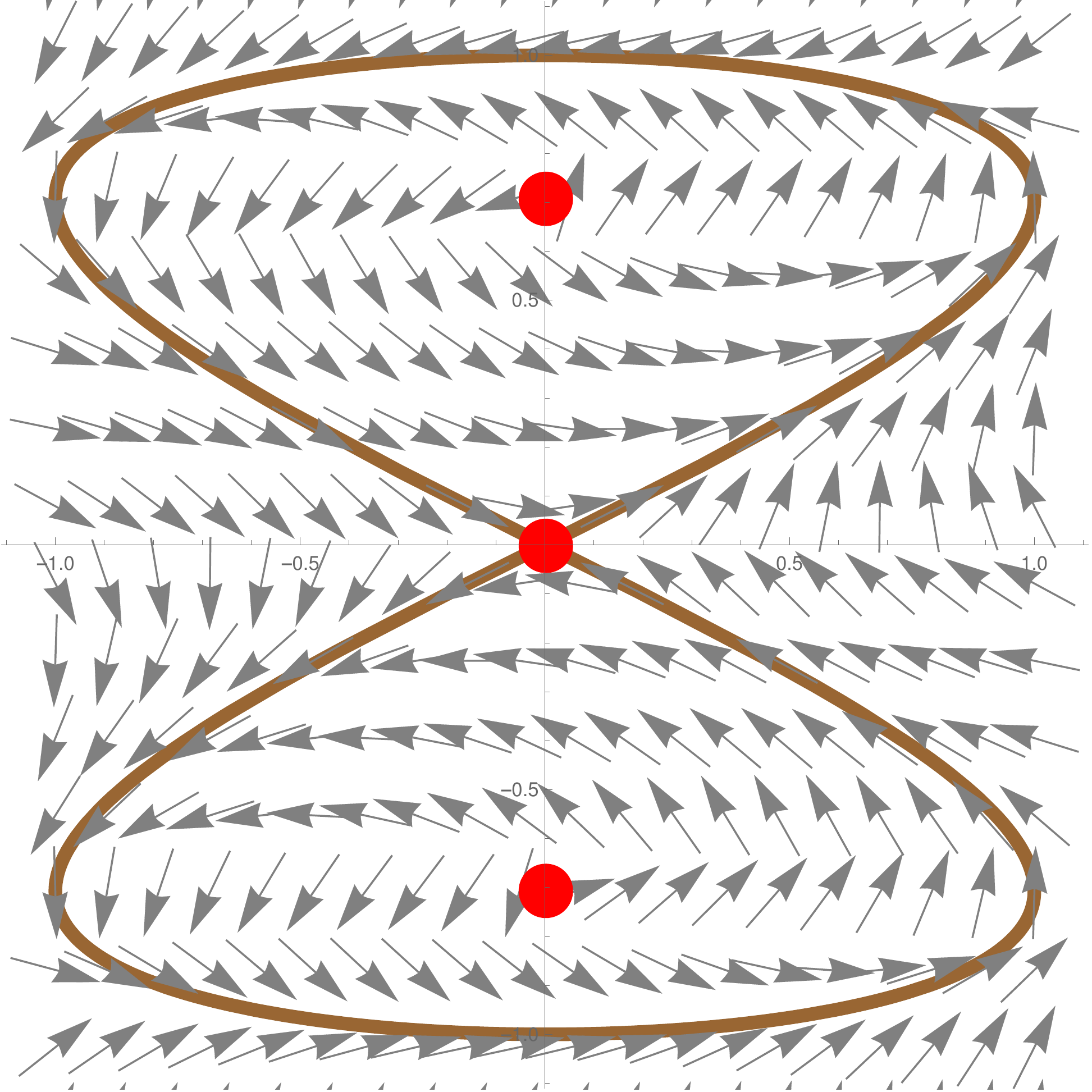} 
		\caption{The normalized vector field for a figure ``8'' path characterized by $\overline{\phi}(x,y) = x^2 - 4 y^2 (1-y^2) = 0$. The three red points are singular points of the vector field in \eqref{eqvf} where normalization is not well defined.}
		\label{fig:eightvf}
		\vspace{-0.5cm}
	\end{figure}
	
	\section{Problem Formulation} \label{sec3}
	We propose to add an additional dimension to the original 2D guiding vector field such that the directional information is ``recorded''. Thus, even at crossing points, the vector field's additional dimension can indicate the preferred direction of motion. 
	%
	
	Starting from the description of the 2D (physical) desired path $\phy{\setp}$ in \eqref{eqpath2d}, the (virtual) desired path $\hgh{\setp}$ in $\mbr[3]$ is naturally characterized by adding an additional constraint (or surface function) in $\mbr[3]$ as follows:
	\begin{equation} \label{path1}
		\hgh{\setp}= \{ p \in \mathbb{R}^3 : \phi_1(p)=0, \; \phi_2(p)=0 \},
	\end{equation}
	where $\phi_1, \phi_2 \in C^2$, $p=(x,y,w) \in \mbr[3]$ and $w$ is the additional dimension of the vector field. This implies that the desired path $\hgh{\setp}$ is the intersection of two surfaces described by $\phi_i=0, i=1,2$. The corresponding 3D guiding vector field $\vf: \mbr[3] \to \mbr[3]$ is \cite{yao2018cdc}:
	\begin{equation} \label{eqgvf}
	\vf(p) = \nabla \phi_1(p) \times \nabla \phi_2(p) - \sum_{i=1}^{2} k_i \phi_i(p) \nabla \phi_i(p).
	\end{equation}
	Here, the latter term $-\sum_{i=1}^{2} k_i \phi_i \nabla \phi_i$ is the sum of signed gradients, providing a direction towards the two surfaces characterized by $\phi_i=0$, and is intuitively called the \emph{converging term}. The first term $\nabla \phi_1 \times \nabla \phi_2$, as the cross product of the two gradients, provides a tangential direction to each surfaces $\phi_i(\xi)=0$, and is intuitively called the \emph{propagation term}. The set of all singular points of the vector field, called the singular set, is formally defined below:
	\begin{equation} \label{setc}
	\hgh{\setc} = \{c \in \mbr[3] : \chi(c)=\boldsymbol{0} \}.
	\end{equation}
	Some mild assumptions  \cite{yao2018cdc} are necessary for the VF-PF problem as listed below. We denote the converging term by $e_{\set{M}}(\cdot) = -\sum_{i=1}^{2} k_i \phi_i(\cdot) \nabla \phi_i(\cdot)$ and define the path-following error vector $e_{\setp}(\cdot)=\big( \phi_1(\cdot), \phi_2(\cdot) \big)$ in view of \eqref{path1}.
	\begin{assump} \label{assump1}
		There are no singular points on the desired path. More precisely, $\hgh{\setc}$ is empty or otherwise there holds $\dist(\hgh{\setc}, \hgh{\setp}) > 0$. 
	\end{assump}
	\begin{assump} \label{assump2}
		 As the norm of the path-following error $\norm{e_{\setp} \big( p(t) \big)}$ approaches zero, the trajectory $p(t)$ approaches the desired path $\hgh{\setp}$. Similarly, as the norm of the converging term $\norm{e_{\set{M}} \big( p(t) \big)}$ approaches zero, the trajectory $p(t)$ approaches the union of the sets $\hgh{\setp} \cup \hgh{\setc}$.	
	\end{assump}
	See \cite{yao2018cdc} for the precise mathematical expressions for Assumption \ref{assump2}. Note that Assumption \ref{assump1} is violated for a self-intersected desired path by Proposition \ref{prop1}, but we will extend it in a higher-dimensional space such that it becomes non-self-intersected. These assumptions are made to avoid pathological consequences and ensure that $\hgh{\setp}$ is a one-dimensional manifold \cite{yao2020auto}, consistent with practical applications. Note that the guiding vector field in \eqref{eqgvf} is nonlinear, and thus the differential equation $\dot{p}(t)=\chi(p(t))$ is generally difficult to analyze. However, under the above assumptions, the integral curves of the 3D vector field only have two outcomes; this is stated in the following lemma \cite{yao2018cdc}.
	\begin{lemma}[Dichotomy Convergence] \label{lemma1}
	The integral curves of the guiding vector field $\vf$ in \eqref{eqgvf} converge to either the desired path $\hgh{\setp}$, or the singular set $\hgh{\setc}$.
	\end{lemma}
	
	Since the ultimate objective is to follow the 2D physical desired path $\phy{\setp}$, we need to project the (virtual) higher-dimensional desired path $\hgh{\setp}$ somehow. More specifically, a linear projection operator $P_a: \mbr[3] \to \mbr[3]$ is defined; this operator is a linear map that can project a vector to the hyperplane orthogonal to a given nonzero vector $a \in \mbr[3]$. With a slight abuse of notation, we use the same symbol for both this linear map and its matrix representation. Therefore,
	\begin{equation} \label{proj}
	P_a = I - \hat{a} \transpose{\hat{a}},
	\end{equation}
	where $I$ is the identity matrix of suitable dimensions and $\hat{a} := a / \norm{a}$. The \emph{projected desired path} and the \emph{projected singular set} are defined below:
	\begin{align}
	\proj{\setp} &= \{ q \in \mbr[3] : q = P_a \,  \xi, \; \xi \in \hgh{\setp} \} \label{pp} \\
	\proj{\setc} &= \{ q \in \mbr[3] : q = P_a \,  \xi, \; \xi \in \hgh{\setc} \}. \label{pc}
	\end{align}
	Therefore, to let the integral curves of the guiding vector field follow a self-intersected 2D desired path, there should be no singular points in the (higher-dimensional) guiding vector field (due to Proposition \ref{prop1}), which also implies the appealing feature of global convergence to the desired path (since $\hgh{\setc}=\emptyset$ in Lemma \ref{lemma1}). To sum up, the problem is formally formulated as follows:
	
	\begin{problem} \label{problem1}
		Given a (possibly self-intersected) physical desired path $\phy{\setp} \subseteq \mbr[2]$, we aim to find a higher-dimensional desired path $\hgh{\setp} \subseteq \mbr[3]$, which satisfies the following conditions\footnote{Topologically, the desired path is one-dimensional, independent of the dimensions of the Euclidean space where it is. However, for convenience, a desired path is called $n$-dimensional if it lives in the $n$-dimensional Euclidean space $\mbr[n]$ and not in any lower-dimensional subspace $\mathcal{W} \subseteq \mbr[n]$.}: 
		
		1) There exist functions $\phi_1, \phi_2 \in C^2$ such that $\hgh{\setp}$ is described by \eqref{path1};
		
		2) The singular set $\hgh{\setc}$ of the corresponding higher-dimensional vector field $\vf:\mbr[3] \to \mbr[3]$ in \eqref{eqgvf} is empty.
		
		3) There exists a projection operator $P_a$ in \eqref{proj} such that 
		the projected desired path $\proj{\setp}$ in \eqref{pp} satisfies $\proj{\setp} = \{(x,y,0) \in \mbr[3]: (x,y) \in \phy{\setp}$\};		
	\end{problem}
	As shown in the problem description, the essential approach is re-designing functions $\phi_i$ corresponding to the new desired path $\hgh{\setp}$ such that the rendered guiding vector field in \eqref{eqgvf} is singularity-free. In the next section, we propose a new idea to seek functions $\phi_i$, and hence a higher-dimensional desired path $\hgh{\setp} \subseteq \mbr[3]$.

	\section{Self-intersected Desired Path Following} \label{sec4}
	
	This section is split into three subsections regarding the solution to Problem \ref{problem1}. Firstly, we introduce the \emph{extended dynamics} related to higher-dimensional vector fields. Secondly, we provide the detailed construction of the 3D (virtual) desired path $\hgh{\setp}$ satisfying the conditions in Problem \ref{problem1}. Once such a 3D desired path is obtained, the 3D vector field will be automatically generated by \eqref{eqgvf}. Thirdly, a control algorithm is designed for a unicycle robot to follow a self-intersected desired path using the 3D vector field.
	
	\subsection{Extended Dynamics} \label{subsec1}
	We introduce the extended dynamics to derive a general result about the integral curves of the projected vector field $P_a \,  \vf$.
	\begin{lemma}[Extended Dynamics] \label{lemma2}
		Let $\vf: \mathcal{D} \subseteq \mbr[3] \to \mbr[3]$ be a vector field that is locally Lipschitz continuous. Suppose $p(t)$ is the unique solution to the initial value problem $\dot{p}(t)=\vf(p(t)), \; p(0)=p_0 \in \mathcal{D}$. Then $(p(t), \proj{p}(t))$, where $\proj{p}(t)=P_a \,  p(t)$ and $P_a$ is the projection operator in \eqref{proj} associated with a given nonzero vector $a \in \mbr[3]$, is the unique solution to the following initial value problem:
		\begin{equation} \label{eq6}
		\begin{cases}
		\dot{p}(t) =\vf(p(t)) & p(0) = p_0\\
		\dot{\proj{p}}(t) = P_a \,  \vf(p(t)) & \proj{p}(0) = P_a \,  p_0,
		\end{cases}
		\end{equation}
		Moreover, if the solution $p(t)$ converges to some set $\mathcal{A} \ne \emptyset \subseteq 
		\mbr[3]$, then the \emph{projected solution} $\proj{p}(t)$ converges to the \emph{projected set}
		$
		\mathcal{A}' = \{ q \in \mbr[3] :  q = P_a \,  \xi, \; \xi \in \mathcal{A} \}.
		$
	\end{lemma}
	\begin{proof}
		Since $\vf$ is locally Lipschitz and $\norm{P_a} = 1$, where $\norm{\cdot}$ is the induced matrix $2$-norm, it follows that $P_a \,  \vf$ is also locally Lipschitz continuous. Therefore, $(p(t), \proj{p}(t))$, where $\proj{p}(t)=P_a \,  p(t)$, is the unique solution to \eqref{eq6}. 
		Fix $t$, then
 $
\dist(\proj{p}(t), \mathcal{A}') = \inf\{ \norm{P_a \, (p(t) - q)} : q \in \mathcal{A} \} 
\le \inf\{ \norm{P_a} \norm{p(t) - q} : q \in \mathcal{A} \}
= \dist(p(t), \mathcal{A}).
$
		Since $p(t)$ converges to $\mathcal{A}$, $\dist(p(t), \mathcal{A}) \to 0$ as $t \to \infty$. Namely, for any $\epsilon>0$, there exists a $T>0$, such that for all $t \ge T$, $\dist(p(t), \mathcal{A}) < \epsilon$; hence $\dist(\proj{p}(t), \mathcal{A}') \le \dist(p(t), \mathcal{A}) < \epsilon$. Therefore, $\dist(\proj{p}(t), \mathcal{A}') \to 0 $ as $t \to \infty$. Thus the projected solution $\proj{p}(t)$ converges to the projected set $\mathcal{A}'$.
	\end{proof}
	\begin{remark}
		Although we only consider the linear projection operator $P_a$ here, note that this is not a restriction, since a nonlinear projection operator $Q: \mbr[3] \to \mbr[3]$ can be similarly defined under some additional assumptions \cite{yao2020singularity}; e.g., the projected vector field should be replaced by $J(Q) \vf$, where $J(Q)$ is the Jacobian matrix function of $Q$ with respect to its arguments, and $J(Q)$ should be locally Lipschitz continuous. For clarity of exposition, we only investigate the linear projection operator, and thus we are able to show the intuitive graphical interpretation of our proposed approach in the sequel.
	\end{remark}
	One of the simplest projection operators $P_a$ is associated with the vector $a=\transpose{(0 , 0 , 1)}$. This can be used to project a three-dimensional vector to a 2D plane by ``zeroing'' the third coordinates. For the problem of path following, we can therefore design a suitable 3D vector field such that the first two components of the integral curves of the 3D vector field follow the projected 2D desired path. The advantage of this method is that singular points can be eliminated in the new vector field and the global convergence to the physical desired path is guaranteed as discussed below.
	
	\subsection{Construction of the 3D Virtual Desired Path}
	Suppose a 2D (physical) desired path $\phy{\setp}$ is parameterized by
	\begin{equation} \label{eq8}
	x = f_1(w), \quad y = f_2(w),
	\end{equation}
	where $w \in \mbr[]$ is the parameter of the path and $f_i \in C^2, i=1,2$. Then, we can simply let 
	\begin{equation} \label{eqsurface1}
	\phi_1(x,y,w) = x - f_1(w), \quad \phi_2(x,y,w) = y - f_2(w),
	\end{equation}
	such that the 3D virtual desired path is described by \eqref{path1} and \eqref{eqsurface1}. Intuitively, the 3D desired path $\hgh{\setp}$ is obtained by \emph{stretching} the 2D desired path $\phy{\setp}$ along the virtual $w$-axis (see Fig. \ref{fig:globallemniscate1}). Thus \emph{the first condition of Problem \ref{problem1} is satisfied}. One can calculate that $\nabla\phi_1 = \transpose{ \big(1 , 0 , -f'_1(w) \big)}$ and $\nabla\phi_2 = \transpose{ \big( 0 , 1 , -f'_2(w) \big) }$, where $f'_i(w) \defeq \dw{f_i(w)}, i=1,2$. Thus 
	\[
	\nabla\phi_1 \times \nabla\phi_2  = \transpose{(f'_1(w), f'_2(w), 1)}.
	\]
	It is interesting to note that the third coordinate of this vector is a constant $1$ regardless of the specific form of the desired path. This means that $\norm{\nabla\phi_1 \times \nabla\phi_2 } \ne 0$ in $\mathbb{R}^3$ globally. A closer examination of the guiding vector field \eqref{eqgvf} reveals that the propagation term $\nabla \phi_1(p) \times \nabla \phi_2(p)$ is orthogonal to the converging term $\sum_{i=1}^{2} -k_i \phi_i(p) \nabla \phi_i(p)$ due to the property of the cross product. Therefore, since $\norm{\nabla\phi_1 \times \nabla\phi_2 } \ne 0$ in $\mathbb{R}^3$ globally, an appealing property is that the vector field $\vf(p) \ne 0$ globally in $\mathbb{R}^3$. This means that the guiding vector field $\vf$ has no singular points. Therefore, \emph{the second condition of Problem \ref{problem1} is met}. Furthermore, the 2D desired path $\phy{\setp}$ is the projection of the 3D (virtual) desired path $\hgh{\setp}$ on the plane $w=0$. Therefore, the projection operator $P_a$ can be to associate with the vector $a=\transpose{(0 , 0 , 1)}$. Thus, \emph{the third condition of Problem \ref{problem1} is also satisfied}. Now we can state the following theorem.
	\begin{theorem} \label{thm1}
		Consider a 2D desired path $\phy{\setp} \subseteq \mbr[2]$ parametrized by \eqref{eq8}. Let $\phi_1$ and $\phi_2$ be chosen as in \eqref{eqsurface1}. Then there are no singular points in the corresponding three-dimensional vector field $\vf:\mbr[3] \to \mbr[3]$.  Let $a=\transpose{(0 , 0 , 1)}$ for the projection operator $P_a$. Suppose the projected solution to \eqref{eq6} is $\proj{p}(t) \defeq \transpose{(x(t), y(t), w(t))}$. Then the 2D trajectory $\proj{p}_2(t) \defeq \transpose{(x(t) , y(t))}$ will globally asymptotically converge to the 2D physical desired path $\phy{\setp}$ as $t \to \infty$.
	\end{theorem}
	\begin{proof}
		From \eqref{eqgvf} and \eqref{eqsurface1}, the 3D vector field $\vf$ is 
		\[
		\resizebox{0.7\hsize}{!}{$
			\vf(x,y,w) = \matr{f_1'(w) - k_1 \phi_1 \\ f_2'(w) - k_2 \phi_2 \\ 1 + k_1 \phi_1 f_1'(w) + k_2 \phi_2 f_2'(w)}.
			$}
		\]
		As discussed before, the singular set $\hgh{\setc} = \emptyset$. According to Lemma \ref{lemma1} and Lemma \ref{lemma2}, together with $\hgh{\setc}=\emptyset$, then the projected solution $\proj{p}(t)$ will globally asymptotically converge to the projected desired path $\proj{\setp}=\{q \in \mbr[3]: q = P_a \,  \xi, \, \xi \in \hgh{\setp} \}$, where $\hgh{\setp}$ is defined by \eqref{path1} and \eqref{eqsurface1}. Since $\transpose{a} \proj{p}(t) = \transpose{a} P_a\, p(t) = 0$, the third coordinate of the projected solution $w(t) \equiv 0$, meaning that the trajectory $\proj{p}(t)$ lies on the $XY$-plane. Therefore, the 2D trajectory $\proj{p}_2(t)\defeq \transpose{(x(t) , y(t))}$ will globally asymptotically converge to the 2D desired path $\phy{\setp}$.
	\end{proof}
	An example to explain Theorem \ref{thm1} is shown in Fig. \ref{fig:globallemniscate}. 
	\begin{remark}
		The additional coordinate $w$ is reminiscent of the time variable $t$ in trajectory tracking algorithms, but there are major differences. In trajectory tracking, a robot tracks a desired trajectory point which moves independently of the robot states along the desired trajectory; namely, the dynamics of time is open loop $\dot{t}=1$.  In our proposed approach, the point $(f_1(w), f_2(w))$ can be roughly regarded as a counterpart of the trajectory point, but note that the dynamics of $w$, i.e., $\dot{w}(t)=\chi_3(p(t))$, where $\chi_3$ is the third entry of the vector field, is dependent on the robot states, and acts in a closed-loop manner. Note that the robot is not tracking the point $(f_1(w), f_2(w))$. The details and a comparison study are presented in Section \ref{sec6}.
	\end{remark}
	\begin{remark} \label{remark1}
		It is important to highlight three advantages of our approach. 
		
		1) (Path Topology) By \emph{stretching} along the virtual coordinate $w$, our approach transforms a possibly self-intersected 2D desired path $\phy{\setp}$ to a non-self-intersected 3D virtual desired path $\hgh{\setp}$, and thus Assumption \ref{assump1} is satisfied;
		
		2) (Singularity Elimination and Global Convergence) \emph{All} singular points-whether they are saddle points or even stable nodes-in the original 2D guiding vector field $\overline{\vf}$ are eliminated in the 3D guiding vector field $\vf$. Due to the singularity-free vector field $\vf$, the \emph{global} convergence to the 3D virtual desired path $\hgh{\setp}$ is guaranteed, so it is for the 2D physical desired path $\phy{\setp}$ by the extended dynamics. We remark that most of the path-following algorithms in the literature can only guarantee local convergence; 
		
		3) (Surface Functions) Our approach facilitates the acquisition of the surface functions $\phi_i$ of which the intersection of the zero-level sets is the 3D desired path $\hgh{\setp}$ (see \eqref{path1}), once a parametric form of the 2D physical desired path $\phy{\setp}$ is known.
	\end{remark}
	\begin{figure}[tb]
		\centering
		{
			\subfigure[]{
				\includegraphics[width=0.48\columnwidth]{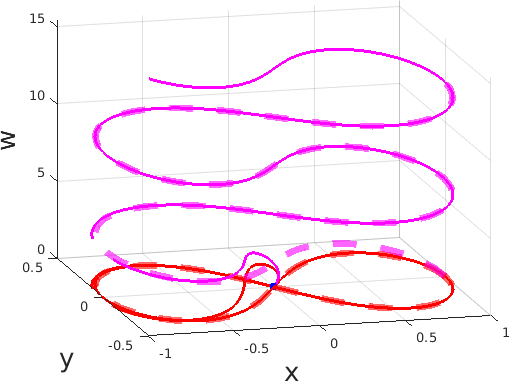}
				\label{fig:globallemniscate1}}%
			\subfigure[]{
				\includegraphics[width=0.48\columnwidth]{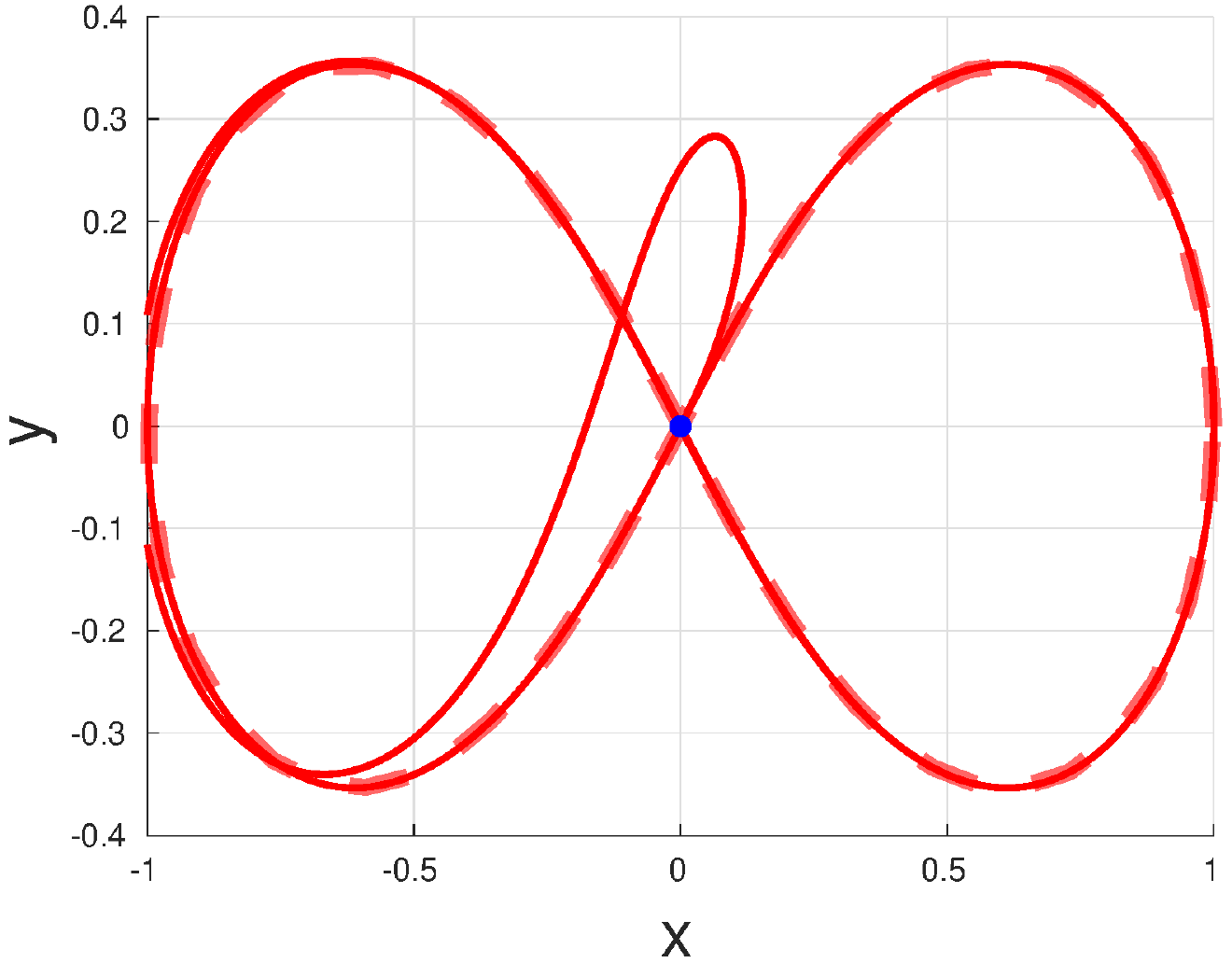}
				\label{fig:globallemniscate3}}%
		}
		\caption{An example of a 2D self-intersected desired path $\phy{\setp}$, of which the parametrization is $x=\cos w / (1+\sin^2 w), y=\sin w \cos w / (1 + \sin^2 w)$. \subref{fig:globallemniscate1} The magenta and red dashed lines are the higher-dimensional desired path $\hgh{\setp}$ and the physical desired path $\phy{\setp}$ respectively. The magenta solid line is the 3D trajectory $p(t)$ of $\dot{p} =  \vf(p)$, and the red solid line is the projected trajectory $\proj{p}_2(t)$ in Theorem \ref{thm1}; \subref{fig:globallemniscate3} The projected trajectory $\proj{p}_2(t)$ in the $XY$-plane. The blue point is the starting point $(0,0)$. Note that the 3D guiding vector field is not depicted as it would render the figure uninterpretable.}
		\label{fig:globallemniscate}
		\vspace{-0.5cm}
	\end{figure}
	
	\subsection{Control Algorithm Design for a Unicycle Robot} \label{subsecc}
	We can exploit the property of the global convergence of integral curves of the guiding vector field to the 2D desired path $\phy{\setp}$. In particular, the control law design principle is to let a mobile robot's orientation eventually aligns with the direction indicated by the guiding vector field. This principle implies that \emph{the guiding vector field only provides guidance signals rather than low-level control commands, and thus it is applicable to any robots whose motions are essentially determined by their orientations}, such as the unicycle model (including the Dubin's car model), the car-like model and the underwater glider model \cite{siciliano2010robotics}. The unicycle robot model is considered as follows:
	\begin{equation} \label{model}
	\dot{x} = v_u \cos \theta \quad \dot{y} = v_u \sin \theta \quad \dot{\theta} = \omega_u,
	\end{equation}
	where $(x,y)$ is the position, $\theta$ is the orientation, $v_u$ is the speed control input and $\omega_u$ is the angular velocity control input. Since the 3D vector field will be used, the generalized 3D velocity vector of the robot needs to be defined as $\dot{p}=\transpose{(\dot{x}, \dot{y}, \dot{w})}$, where $(\dot{x}, \dot{y})$ is the \emph{actual} velocity of the robot as defined in \eqref{model} and $\dot{w}$ is the \emph{virtual} velocity in the additional coordinate that is to be determined later. The control inputs $v_u$, $\omega_u$ and the virtual velocity $\dot{w}$ will be designed such that the orientation and the length of the generalized velocity $\dot{p}=\transpose{(\dot{x}, \dot{y}, \dot{w})}$ will asymptotically become identical to the scaled vector field $s \hat{\vf}$ (i.e., $\dot{p} \to s \hat{\vf}$ as $t \to \infty$), where $s$ is a given positive constant. 
	
	\begin{theorem}
		Suppose the robot model in \eqref{model} is considered, and a 2D parameterized desired path $\phy{\setp} \subseteq \mbr[2]$ is given. The corresponding 3D vector field $\vf: \mbr[3] \to \mbr[3]$ is constructed in Section \ref{sec4}. Assume that the vector field satisfies $\vf_1(p)^2 + \vf_2(p)^2 \ne 0$ for $p \in \mbr[3]$, where $\vf_i$ denotes the $i$-th entry of the vector field $\vf$.  Let
		\begin{subequations} \label{eqcontrollaw}
			\begin{align}
			\dot{w} &= s \hat{\vf}_3,  \label{eqw} \\
			v_u &= s \norm{\vf^p}, \label{eqinput1}\\ 
			\omega_u &= \dot{\theta}_d - k_\theta \transpose{\hat{h}} E \hat{\vf^p}, \label{eqinput2} \\
			\dot{\theta}_d &= \left(\frac{-1}{\norm{\vf^p}} \transpose{\hat{\vf^p}} E J(\vf^p) \dot{p}  \right), \label{thetaddot}
			\end{align}
		\end{subequations}
		where $\hat{(\cdot)}$ is the normalization operator, $s$ and $k_\theta$ are positive constants, $h =\transpose{(\cos \theta, \sin \theta)}$, $\vf^p = \transpose{(\hat{\vf}_1, \hat{\vf}_2)}$,  $E=\left[\begin{smallmatrix}0 & -1 \\ 1 & 0\end{smallmatrix}\right]$, $J(\vf^p)$ is the Jacobian matrix of $\vf^p$ with respect to the generalized position $p \in \mbr[3]$ and $\dot{p}=\transpose{(\dot{x}, \dot{y}, \dot{w})}$ is the generalized velocity. Denote the angle directed from $\hat{\vf^p}$ to $\hat{h}$ by $\beta \in (-\pi, \pi]$. If the initial angle $\beta(0) \in (-\pi, \pi)$, then the generalized velocity $\dot{p}$ will converge asymptotically to the scaled vector field $s \hat{\vf} = s \transpose{(\hat{\vf}_1, \hat{\vf}_2, \hat{\vf}_3)}$ (i.e., $\beta(t) \to 0$). Furthermore, the actual robot trajectory $(x(t), y(t))$ will converge to the physical desired path $\phy{\setp} \subseteq \mbr[2]$ asymptotically as $t \to \infty$.
	\end{theorem}
	\begin{proof}
		We define two vectors in $\mbr[2]$; i.e., $h' \defeq v_u \transpose{(\cos \theta, \sin \theta)}$ and $g' \defeq s \transpose{(\hat{\vf}_1, \hat{\vf}_2)}$. Also define the error (difference) between the generalized velocity $\dot{p}$ and the scaled vector field $s \hat{\vf}$ as below:
		\[
			e_{ori}(t)=\dot{p}-s\hat{\vf} = \matr{v_u \cos\theta - s \hat{\vf}_1 \\ v_u \sin\theta - s \hat{\vf}_2\\ 0} = \matr{h' - g' \\ 0} \in \mbr[3],
		\]
		where the last entry $0$ is due to \eqref{eqw}. Thus one only needs to focus on the first two entries of $e_{ori}(t)$. Note that $\norm{h'}=s \sqrt{\hat{\vf}_1^2 + \hat{\vf}_2^2}=\norm{g'}$,	thus it is possible to proceed to show that $e_{ori}(t) \to 0$ asymptotically. In particular, it suffices to show that the orientation of $h'$ asymptotically aligns with that of $g'$. Note that $\hat{h}=\hat{h'}=\transpose{(\cos \theta, \sin \theta)}$ and $\hat{\vf^p}=\hat{g'}=\transpose{(\hat{\vf}_1, \hat{\vf}_2)} / \sqrt{\hat{\vf}_1^2 + \hat{\vf}_2^2}=\transpose{(\vf_1, \vf_2)}/\sqrt{\vf_1^2+\vf_2^2}$. Let $e'_{ori}= \hat{h} - \hat{\vf^p}$. Choose the Lyapunov function candidate $V = 1/ 2 \, \transpose{e'}_{ori} e'_{ori}$, and its time derivative is
		\begin{equation} \label{eqvdot}
		\begin{split}
		\dot{V} = \transpose{\dot{e}'}_{ori} e'_{ori} 
		&= \transpose{( \dot{\theta} E \hat{h}  - \dot{\theta}_d E \hat{\vf^p})} (\hat{h} - \hat{\vf^p})  \\
		&= (\dot{\theta} - \dot{\theta}_d) \transpose{\hat{h}} E \hat{\vf^p} \\
		&\overset{\eqref{eqinput2}}{=} -k_\theta (\transpose{\hat{h}} E \hat{\vf^p})^2 \le 0.
		\end{split}
		\end{equation}
		The second equation makes use of the identities: $\dot{\hat{h}} = \dot{\theta} E \hat{h}$ and $\dot{\hat{\vf^p}} = \dot{\theta}_d E \hat{\vf^p}$, where $\dot{\theta}_d$ is defined in \eqref{thetaddot}. The third equation is derived by exploiting the facts that $\transpose{E}=-E$ and $\transpose{a} E a = 0$ for any vector $a \in \mbr[2]$. Note that $\dot{V}=0$ if and only if the angle difference between $\hat{h}$ and $\hat{\vf^p}$ is $\beta=0$ or $\beta=\pi$. Since it is assumed that the initial angle difference $\beta(t=0) \ne \pi$, it follows that $\dot{V}(t=0)<0$, and thus there exists a sufficiently small $\epsilon>0$ such that $V(t=\epsilon)<V(t=0)$. It can be shown by contradiction that $|\beta(t)|$ is monotonically decreasing with respect to time $t$. Suppose there exist $0<t_1<t_2$ such that $|\beta(t_1)|<|\beta(t_2)|$. It can be calculated that $V(t)=1-\cos\beta(t)$, and thus $V(t_1)<V(t_2)$, contradicting the decreasing property of $V$. Thus $|\beta(t)|$ is indeed monotonically decreasing. By \eqref{eqvdot}, $|\beta(t)|$ and $V(t)$ tends to $0$, implying that the generalized velocity $\dot{p}$ will converge asymptotically to the scaled vector field $s \hat{\vf}$. Therefore, the generalized trajectory $(x(t), y(t), w(t))$ will converge to the higher-dimensional desired path $\hgh{\setp}$ constructed in Section \ref{sec4} \cite{yao2020auto}. Then by Theorem \ref{thm1}, the actual robot trajectory $(x(t), y(t))$ (i.e., $\proj{p}_2(t)$ in Theorem \ref{thm1})  will converge to the physical desired path $\phy{\setp} \subseteq \mbr[2]$ asymptotically as $t \to \infty$.
	\end{proof}
	\begin{remark}
		A significant feature of this algorithm is the introduction of a virtual coordinate and its dynamics. These do not correspond to any physical quantities in practice, but it is indispensable to form the extended dynamics as shown in Lemma \ref{lemma2}. By setting the third entry of the generalized velocity to be the same as the corresponding entry of the scaled vector field, it has been ensured that the third entry is always ``aligned'' and thus the algorithm mainly deals with the alignment of the other two entries with the scaled vector field. 
	\end{remark}	
	\begin{remark}
		We explain why we do not consider constraints related to the control inputs or the desired path curvatures: As the mobile robot is globally guided by the vector field, it can always re-orient its heading towards the desired path even if it temporarily deviates from the desired path due to these constraints. The fixed-wing aircraft experiment in \cite{yao2020singularity} verifies the effectiveness of the control law, but a more detailed theoretical analysis is left for future work.
	\end{remark}
	
	\section{Experiment} \label{sec5}
	An e-puck robot \cite{mondada2009puck} is employed to follow a self-intersected desired path: the projection of a trefoil knot that is parameterized by  
	$
	x =  \cos(0.02 w)(80\cos(0.03 w) + 160) + 600, \;
	y =  \sin(0.02 w)(80\cos(0.03 w) + 160) + 350, 
	$
	where $w$ is the parameter of the desired path. We use pixels as the distance unit. The robot has on top a data matrix, which is recognized by a overhead camera to obtain the robot's position and orientation. The camera is connected to a computer, where the control algorithm runs. Then, the linear velocity and angular velocity control inputs $v_u$ and $\omega_u$ are transmitted from the computer via a Bluetooth module at a fixed frequency of $20$ Hz to the robot. The initial configuration of the robot is: $(x(0), y(0), \theta(0))=(923, 545, \pi)$ and the initial value of the additional coordinate is $w(0)=0$. The rest of the parameters are chosen to be $s=10, k_1=0.5, k_2=0.2, k_e=50$. The implementation procedure is listed in Algorithm \ref{algo1}, and the experimental results are shown in Fig. \ref{fig_exp}. At the time $t=87s$, the robot was manually moved away from the desired path. However, after that, it headed towards the desired path again to reduce the path-following error. As shown in Fig. \ref{fig:exp2}, the norm of the path-following error characterized by $\norm{(\phi_1, \phi_2)}$ decreased significantly and fluctuated around zero eventually in the presence of sensor noise.
	%
	
	\begin{algorithm}
		\caption{Proposed VF-PF control algorithm}
		\begin{algorithmic} \label{algo1}
			\REQUIRE Parametrization of the desired path in \eqref{eq8}
			\STATE Initialized the virtual coordinate $w(0) = w_0$.
			\WHILE{stop signal not received}
			\STATE Obtain robot states $(x,y,\theta)$ and velocity $(\dot{x}, \dot{y})$.
			\STATE Calculate 3D vector field $\vf$ by \eqref{eqgvf} and \eqref{eqsurface1}.
			\STATE Calculate $\dot{w}$ by \eqref{eqw} and let $\dot{p}=\transpose{(\dot{x},\dot{y},\dot{w})}$.
			\STATE Calculate $v_u$ and $\omega_u$ by \eqref{eqinput1} and \eqref{eqinput2}.
			\STATE Apply the control inputs $v_u$ and $\omega_u$ to the robot.
			\STATE Update $w \leftarrow w + \dot{w} \; \Delta t$, where $\Delta t$ is the iteration lapses.
			\ENDWHILE
		\end{algorithmic}
	\end{algorithm}
	\section{Discussion: Path or trajectory tracking?} \label{sec6}
	In this section, we discuss whether our presented algorithm falls into the category of path tracking (i.e., path following) or trajectory tracking. While our guiding vector field is the standard output for the path tracking approach, we argue that it is also an extension of the trajectory tracking approach. Therefore, our approach combines and extends elements from both algorithms.
	\begin{figure}[tb]
	\centering
	{
		\subfigure[]{
			\includegraphics[width=0.5\linewidth]{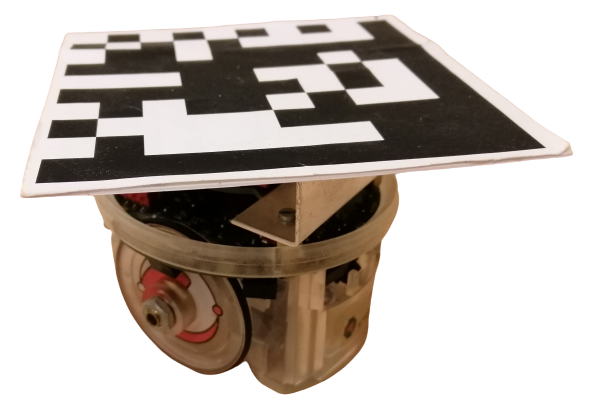}
			\label{fig:exp0}}
		\subfigure[]{
			\includegraphics[width=0.5\linewidth]{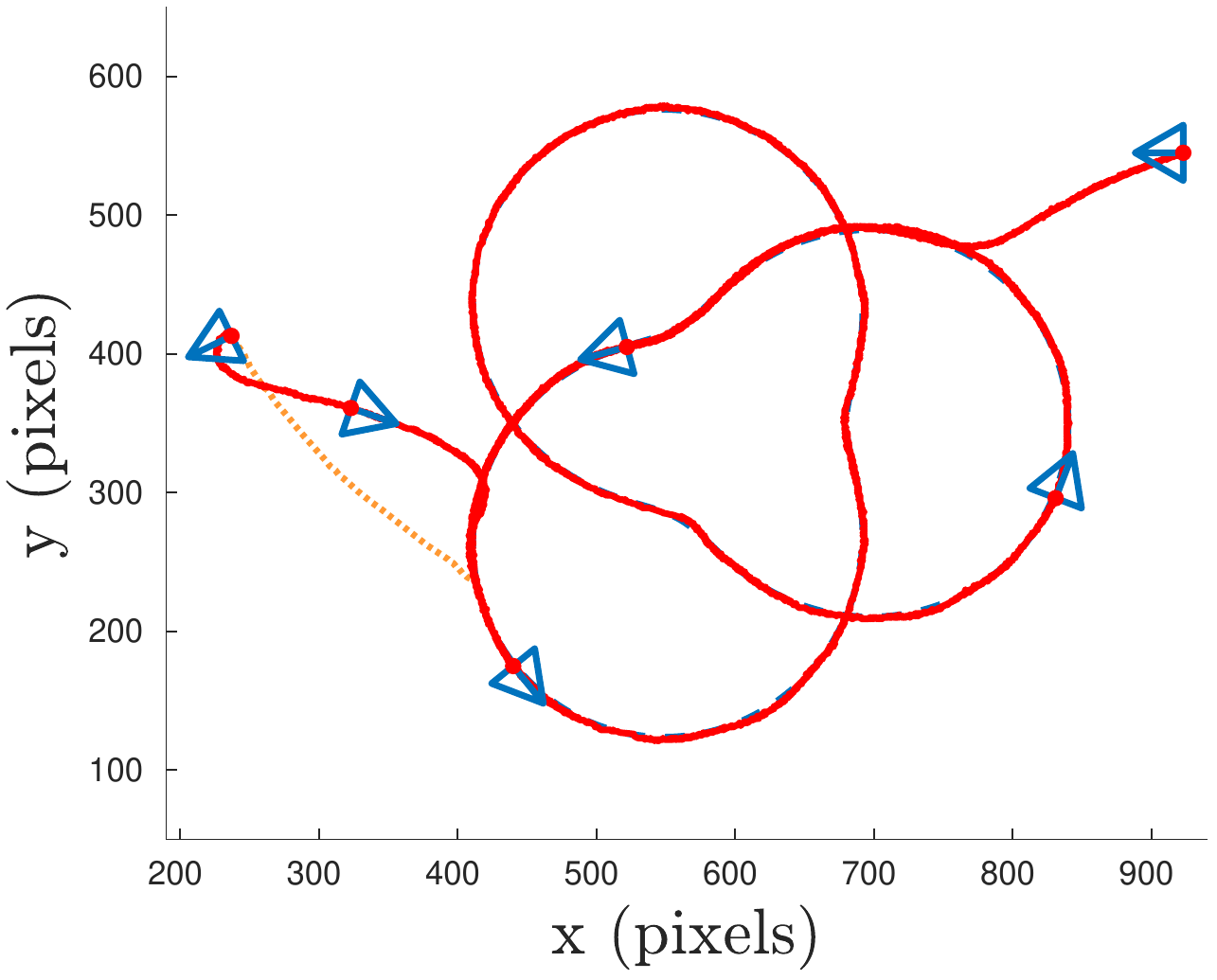}
			\label{fig:exp1}}%
		\subfigure[]{
			\includegraphics[width=0.47\linewidth]{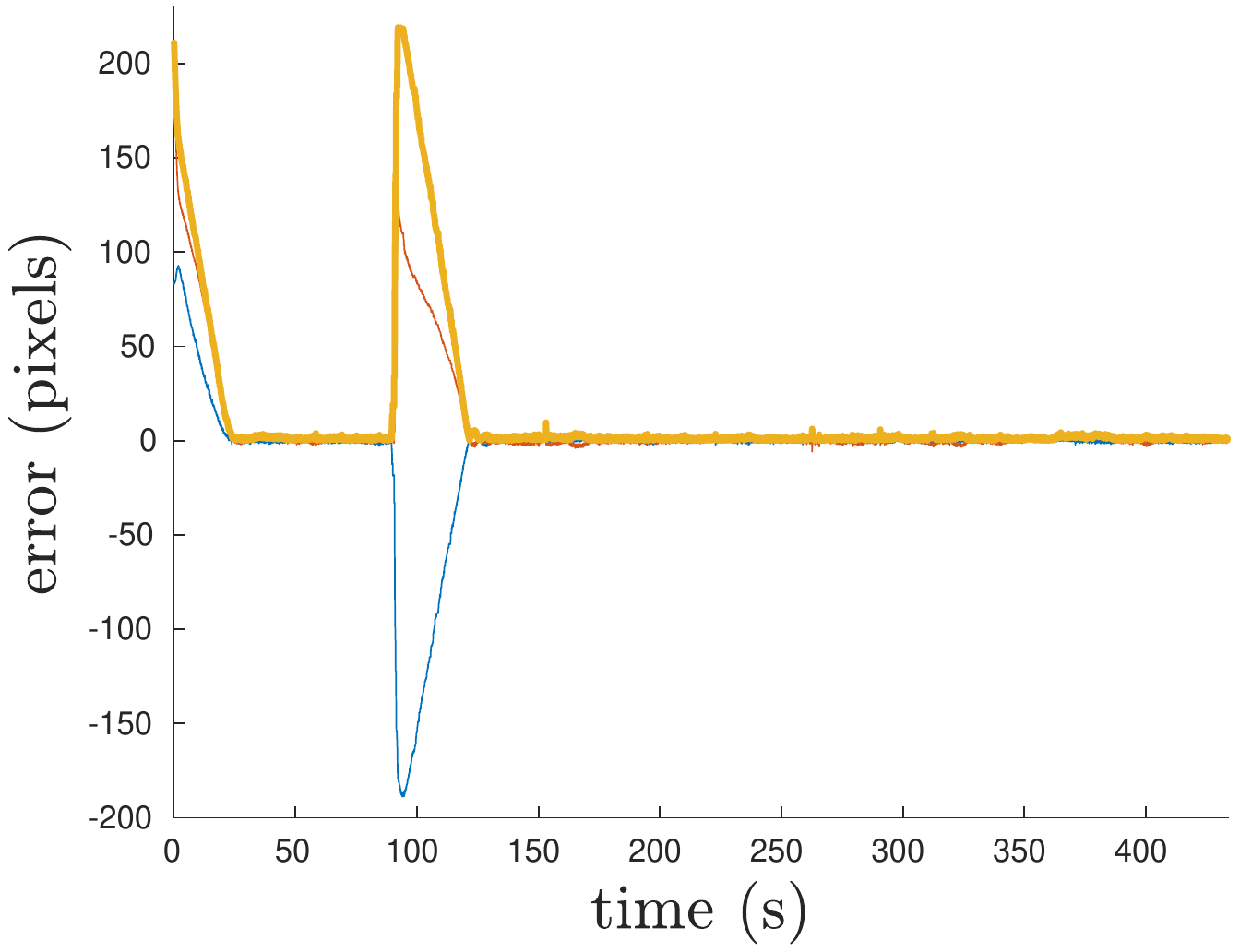}
			\label{fig:exp2}}%
	}
	\vspace{-0.5cm}
	\caption{Experiment results. \subref{fig:exp0} The e-puck robot;  \subref{fig:exp1} Visualization of the experiment data. The red solid line is the actual robot trajectory, and the blue dashed line is the desired path (mostly covered by the robot trajectory). Note that at $87s$ the robot was manually moved to the leftmost position in the figure (indicated by the orange dashed line). The triangles represent the robot positions at different time instants, where the medians of the triangles pointing from the edge to the vertex indicate the robot headings; \subref{fig:exp2} The path-following errors $\phi_1$, $\phi_2$ and $\norm{(\phi_1, \phi_2)}$ represented by blue, red and yellow lines respectively.} 
	\label{fig_exp} 
	\vspace{-0.5cm}
\end{figure} 

	%
	
	Firstly, note that a trajectory tracking algorithm is designed around the desired trajectory that specifies desired positions, velocities, etc. at specific time instants.  The ``drawn path'' consists of a sequence of \emph{desired trajectory points}; namely, the time evolution of the desired position. By contrast, our proposed algorithm treats the desired path as a geometric object. For VF-PF algorithms, a parametrization of the desired path is not required to generate a ``desired trajectory point''. However, we exploit the parametrization in Theorem \ref{thm1} to facilitate the expressions of the surface functions $\phi_i$, of which the intersection of the zero-level sets is precisely the (3D) desired path $\hgh{\setp}$. Note that the subsequent derivation of the higher-dimensional guiding vector field depends merely on $\phi_i$, independent of the specific parametrization of the physical desired path $\phy{\setp}$.
	
	Secondly, the additional coordinate $w$ in our proposed VF-PF algorithm seems similar to but in fact differs from the time $t$ in trajectory tracking. In trajectory tracking, we remind that $\dot t = 1$; i.e., the desired trajectory $r(t)$ is prescribed, and it evolves as time elapses in an \emph{open-loop} manner independent of the states of the robot. However, in our approach, the dynamics of the additional coordinate $\dot{w}(t)=s \hat{\vf}_3(p(t))$ in \eqref{eqw} are in the closed-loop with the states of the robot (i.e., $p(t)$). Consequently, we improve the performance of our algorithm, e.g., under noisy measurements, with respect to the standard trajectory tracking approach. This claim is further justified by numerical experiments and theoretical studies as follows\footnote{Supplementary material: \url{http://tiny.cc/cdc20_yao}.}.
	
	Using the unicycle robot model in \eqref{model}, we compare the proposed VF-PF algorithm in Section \ref{sec4} with the nonlinear trajectory tracking algorithm introduced in the classic monograph \cite[p. 506]{siciliano2010robotics}. For path following, we choose the desired path as the projection of a Lissajous knot \cite{bogle1994lissajous} parameterized as follows: $x =  250 \, \cos(0.06 w + 0.1) + 600, \;
	y =  250 \, \cos(0.08 w + 0.7) + 350 $,
	where $w \in \mbr[]$ is the parameter. Then we use \eqref{eqgvf}, \eqref{eqsurface1} to create the 3D vector field and \eqref{eqcontrollaw} to calculate the control inputs $v_u$ and $\omega_u$. The control gains are $k_1, k_2$ for the converging term of the vector field and $k_\theta$ for the angular control input $\omega_u$. For trajectory tracking, the prescribed trajectory is obtained by replacing the path parameter $w$ by time $t$. Note that the dynamics of $t$ are trivially the \emph{open-loop} $\dot t = 1$, while the dynamics of the path parameter $w$ are in the \emph{closed-loop} \eqref{eqw}. For ease of explanation, the desired trajectory is denoted by $(x_d(t), y_d(t))$, and the feasible desired heading $\theta_d(t)$ is computed from $(x_d(t), y_d(t))$ \cite[p. 503]{siciliano2010robotics}. The control inputs for the trajectory tracking algorithm are as follows:
	\begin{subequations} \label{eqtrajtrack}
	\begin{align}
	e' &\defeq \left[\begin{matrix}e'_1 \\ e'_2 \\ e'_3\end{matrix}\right] = \left[\begin{matrix}\cos\theta & \sin\theta & 0 \\ - \sin\theta & \cos\theta & 0 \\ 0 & 0 & 1\end{matrix}\right] \left[\begin{matrix}x_d - x \\ y_d - y \\ \theta_d - \theta \end{matrix}\right]  \\
	u_1 &= - k^{trj}_{1} e'_1 , \;
	u_2 = -k^{trj}_2 v_d e'_2 \sin e'_3 / e'_3 - k^{trj}_3 e'_3 \\
	v_u &= v_d \cos e'_3 - u_1, \;\; 	\omega_u = \omega_d - u_2,
	\end{align}
	\end{subequations}
	where $v_d(t) = \sqrt{\dot{x}_d^2(t) + \dot{y}_d^2(t)}$, $\omega_d(t)=\big(\ddot{y}_d(t) \dot{x}_d(t) - \ddot{x}_d(t)\dot{y}_d(t) \big) / \big(\dot{x}_d^2(t) + \dot{y}_d^2(t) \big)$, and $k^{trj}_1, k^{trj}_2, k^{trj}_3$ are positive control gains. In both numerical simulations, the initial positions ($x(0), y(0)$) and orientations $\theta(0)$ of the robot are the same. For path following, the initial value of the additional coordinate is $w(0)=0$, while correspondingly for trajectory tracking, the initial time instant is $t=0$. In addition, we choose the same control gains for these two simulations; namely, $k_1=k_2=k^{trj}_1=k^{trj}_2=0.05, k_\theta=k^{trj}_3=1$.
	\begin{figure}[tb]
		\centering
		{
			\subfigure[]{
				\hspace*{-0.5cm}
				\includegraphics[width=0.52\columnwidth]{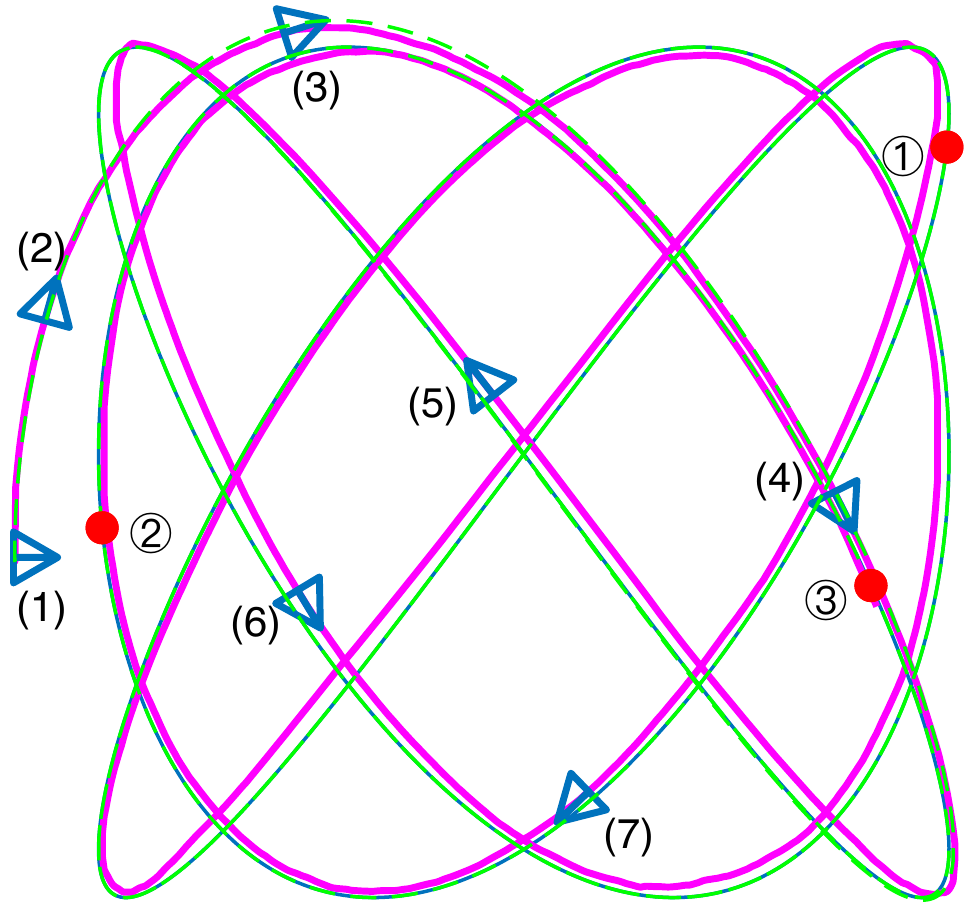}
				\label{fig:sim1}}%
				\hspace*{-0.2cm}
			\subfigure[]{
				\includegraphics[width=0.52\columnwidth]{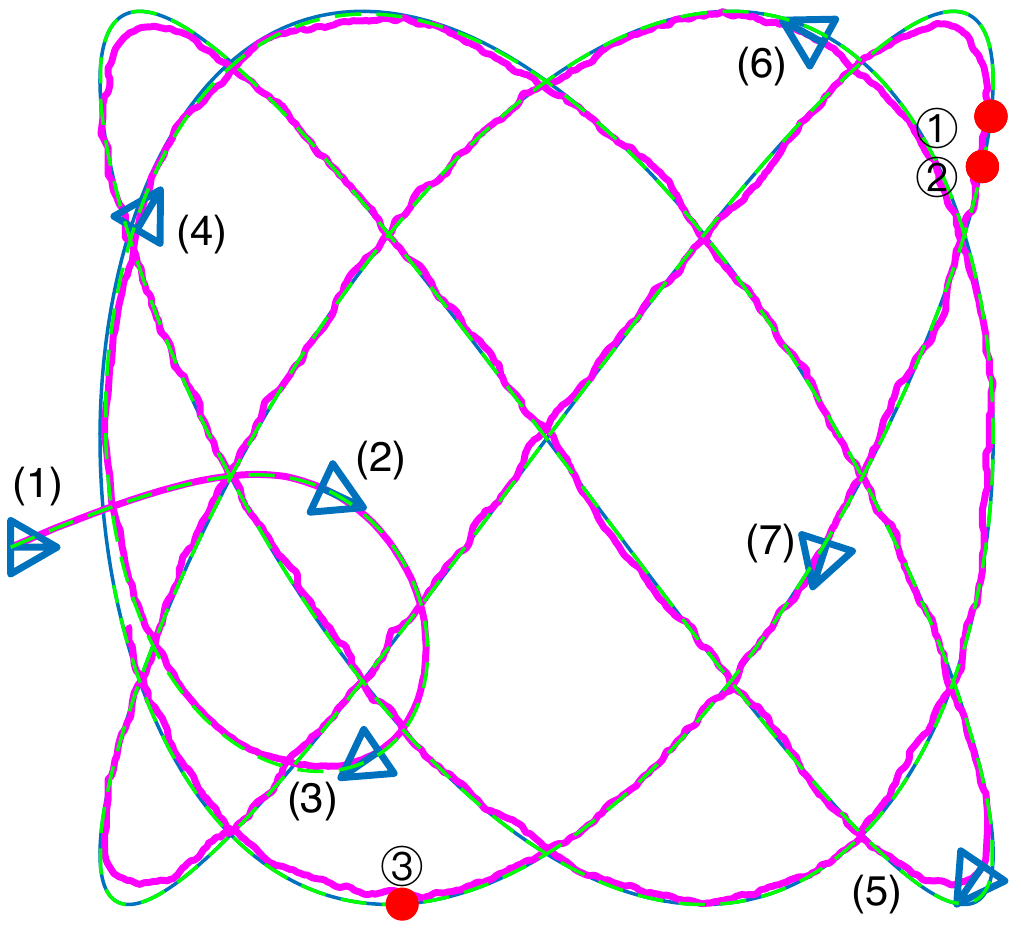}
				\label{fig:sim2}}%
		}
		\caption{\subref{fig:sim1} and \subref{fig:sim2} are simulation results with the proposed VF-PF algorithm and the nonlinear trajectory tracking algorithm respectively. The blue solid lines are the desired paths/trajectories; the green dash lines and magenta solid lines are the robot trajectories without white noise and with white noise added to the perceived robot positions respectively. The small blue triangles with labels represent a sequence of robot positions in the case with noise. The red points with labels are the \emph{guiding points} for path following and \emph{desired trajectory points} for trajectory tracking at $t=0, 2, 30$ s respectively.}
		\label{fig_sim} 
		\vspace{-0.5cm}
	\end{figure}

	If the measurements of the robot positions are accurate (i.e., no white noise is added), then both algorithms enable the robot to follow/track the desired path/trajectory accurately (see the green dash lines in Fig. \ref{fig_sim}). However, by checking the robot positions labeled by $(1), (2)$ and $(3)$ in Fig. \ref{fig:sim1} and Fig. \ref{fig:sim2}, one notes that the robot's initial behaviors are quite different. In the trajectory tracking case, the robot revolves before aligning with the desired path, while in the VF-PF case, the robot moves more ``naturally'': it ``aligns'' with the desired path immediately and reduce the distance to the desired path later. The ``unnatural'' movement for the trajectory tracking algorithm is due to the ``open-loop'' dynamics of the desired trajectory point. In the beginning, the desired trajectory point $(x_d(0), y_d(0))$ is far away from the robot's initial position (see the red point labeled by \ding{172} in Fig. \ref{fig:sim2}), and thus the robot needs to approach it to reduce the tracking error. Since the movement of the desired trajectory point is independent of the robot position, the robot needs to steer its heading continually as the trajectory point moves along the desired path; the desired trajectory points at $t=2$ s and $t=30$ s are illustrated as two red points labeled by \ding{173} and \ding{174} respectively in Fig. \ref{fig:sim2}. In view of \eqref{eqsurface1}, $(f_1(w(t)), f_2(w(t)))$ is termed the \emph{guiding point} as the counterpart of the desired trajectory point. By contrast, although the guiding point $(f_1(w(t)), f_2(w(t)))$ is also on the far right side at $t=0$ (the same as the desired trajectory point), it ``rushes'' to the vicinity of the robot's initial position shortly after $2$ seconds (the red point labeled by \ding{173} in Fig. \ref{fig:sim1}), activated by the closed-loop dynamics  $\dot{w}(t)=s \hat{\vf}_3(p(t))$ and the converging property of the integral curves of the vector field. This enables the robot to approach the ``nearest point'' of the desired path, avoiding the unnatural revolving behavior initially.
	
	To experimentally compare the differences of the two algorithms under artificial sensor noise, we intentionally add a significant amount of band-limited white noise (power: 10, sampling time: 0.1 s) to the positions perceived by the robot. As seen from the magenta lines in Fig. \ref{fig_sim}, the robot trajectory for the path following algorithm is visually more smooth than that for the trajectory tracking algorithm. This is partly attributed to the propagation term in \eqref{eqgvf}, which guarantees that there is always a nonzero tangential term to the level curves of the desired path, not being affected by the converging term due to orthogonality (see the Appendix for the theoretical analysis). This property makes our algorithm appealing for fixed-wing aircraft motion control in the presence of sensor noise: By tuning the gains of the vector field such that the propagation term dominates, we can guarantee that the control will not demand the UAV's heading to vibrate intensively, or even move backwards undesirably. By contrast, there is no tangential term in the trajectory tracking algorithm, and thus the robot heading vibrates persistently as the perceived robot position varies. In some extreme cases, the robot might move backward if the noisy perceived robot position overtook the desired trajectory point.	

	\section{Conclusion and Future Work} \label{sec7}
	We have proposed a 3D singularity-free guiding vector field to enable a robot to follow a self-intersected desired path in 2D by re-designing the surface functions. Moreover, we have rigorously proven, and experimentally validated using a unicycle robot, the global convergence of the guiding vector field's integral curves to self-intersected desired paths. Three advantages of our approach are summarized in Remark \ref{remark1}, regarding the path topology, singularity elimination and global convergence, and the surface functions. We have also compared our proposed VF-PF algorithm with a nonlinear trajectory tracking algorithm and the existing VF-PF algorithms, and concluded that our approach is a significant extension. Note that our proposed algorithm is not restricted to 2D desired paths; in fact, 3D or even higher-dimensional desired paths (possibly in robot configuration spaces) are still applicable using the same idea \cite{yao2020singularity}. In addition, we emphasize the potential of utilizing our approach to create singularity-free guiding vector fields for \emph{global} robot navigation in obstacle-populated environments \cite{rimon1992exact,yao2019integrated}.
	\appendices
	\section{Theoretical analysis for Section \ref{sec6}}
	The vector field \eqref{eqgvf} has two terms as below:
	\begin{equation*}
	\resizebox{.8\hsize}{!}{$
		\vf(p) = \underbrace{\nabla \phi_1(p) \times \nabla \phi_2(p)}_{\tau(p)}  +  \underbrace{\sum_{i=1}^{2} -k_i \phi_i(p) \nabla \phi_i(p)}_{\iota(p)},
		$}
	\end{equation*}
	where the propagation term $\tau(p)$ provides a propagation direction along the desired path, and the converging term $\iota(p)$ ``pushes'' the trajectory towards the desired path. In the presence of considerable noise, the converging term $\iota(p)$ affects greatly the ``smoothness'' of the actual trajectory as the perceived robot position fluctuates. 
	However, the propagation term $\tau(p)$ alleviates this negative effect, by always providing a forward direction.  In addition, in the vicinity of the desired path, $\phi_i \approx 0$, and hence the converging term $\iota(p) \approx 0$, and the propagation term $\tau(p)$ dominates the robot motion. In applications where backward motion is not desirable (e.g., for fixed-wing UAVs), one can manually add a gain to the propagation term $\tau(p)$ such that it always dominates the converging term (or lower the gain for the converging term).  Therefore, the robot trajectory is less sensitive to the noise perturbation.
	
	In fact, it can be proved that our approach has a robustness property due to the closed-loop dynamics of $w$; precisely, the path-following error dynamics is locally input-to-state stable. Denote the disturbance by $d$ and define a set $\mathcal{E}_\alpha \defeq \{p \in \mbr[3]: \norm{(\phi_1(p), \phi_2(p))} < \alpha \}$ with a positive constant $\alpha$. Due to the noise perturbation on the perceived robot position, it is further introduced into the vector field. Hence the perturbed system is:
	\begin{equation} \label{eq_perturb}
	\dot{p}(t)=\vf(p(t))+d(t),
	\end{equation}
	where $d: \mathbb{R}_{\ge 0} \to \mathbb{R}^n$, is assumed to be a bounded and piecewise continuous function of time $t$ for all $t\ge0$. It can be shown that the path-following error dynamics
	$
	\dot{e}(t) = \transpose{N}(p(t)) (\vf(p(t)) + d(t))
	$
	is locally input-to-state stable. 
	\begin{prop}(GVF Robustness) \label{prop_iss}
		Consider the perturbed system \eqref{eq_perturb}. If $|f_i'(w)|$ is upper bounded in $\mathcal{E}_\alpha$ for $i=1,2$, then the error dynamics is locally input-to-state stable. Namely, there exists positive constants $\delta$ and $r$, such that if the initial condition $p(0) \in \mathcal{E}_\delta$ and the disturbance is uniformly bounded by $r$ (i.e., $\norm{d(t)} \le r$ for $t \ge 0$), then the path-following error $\norm{e(t)}$ is uniformly ultimately bounded. In addition, if the disturbance is vanishing; that is, $\norm{d(t)} \xrightarrow{t\to\infty}0$, then the norm of the path-following error $\norm{e(t)}$ is also vanishing. 
	\end{prop}
	\begin{proof}
		The proof is similar to \cite[Theorem 3]{yao2020auto}.
	\end{proof}
	
	
	\section*{Acknowledgments}
	We thank Dr. Zhiyong Sun for his valuable feedback, and Bohuan Lin for the fruitful discussions.
	
	\bibliographystyle{IEEEtran}
	\bibliography{ref}

\end{document}